\documentclass[a4paper,11pt]{article}

\usepackage{graphicx}% http://ctan.org/pkg/graphicx

\usepackage{fullpage}
\usepackage{libertine}
\usepackage{color}

\usepackage[ruled,vlined]{algorithm2e}
\SetArgSty{textrm}

\usepackage{amsmath,amsfonts,amssymb,amsthm}

\usepackage[breaklinks=true]{hyperref}
\usepackage[svgnames]{xcolor}
\usepackage[capitalise,nameinlink]{cleveref}
\hypersetup{colorlinks={true},linkcolor={DarkBlue},citecolor=[named]{DarkGreen}}

\usepackage{natbib}%\usepackage[numbers]{natbib}

\usepackage{subcaption}

\usepackage{tikz}  %needed for \textcolor as well as tikz
\usetikzlibrary{arrows}
\usetikzlibrary{patterns,snakes}
\usetikzlibrary{decorations.shapes}
\tikzstyle{overbrace text style}=[font=\tiny, above, pos=.5, yshift=5pt]
\tikzstyle{overbrace style}=[decorate,decoration={brace,raise=5pt,amplitude=3pt}]
\usetikzlibrary{shapes.geometric}

\newtheorem{theorem}{Theorem}[section]
\newtheorem{corollary}[theorem]{Corollary}

\theoremstyle{definition}
\newtheorem*{comment*}{Comment}
\newtheorem{example}[theorem]{Example}

\newcommand{\cost}{\text{cost}}

\newcommand{\SC}{\text{SC}}
\newcommand{\MC}{\text{MC}}

\newcommand{\bw}{\mathbf{w}}
\newcommand{\bo}{\mathbf{o}}
\newcommand{\by}{\mathbf{y}}

\title{\bf On Truthful Constrained Heterogeneous \\ Facility Location with Max-Variant Cost}

\usepackage{authblk}
\author[1]{Mohammad Lotfi}
\author[2]{Alexandros A. Voudouris}

\affil[1]{ Sharif University of Technology, Iran}
\affil[2]{University of Essex, UK}

\date{}

\begin{document}

\allowdisplaybreaks

\maketitle

\begin{abstract}
We consider a problem where agents have private positions on a line, and public approval preferences over two facilities, and their cost is the maximum distance from their approved facilities. The goal is to decide the facility locations to minimize the total and the max cost, while incentivizing the agents to be truthful. We design a strategyproof mechanism that is simultaneously $11$- and $5$-approximate for these two objective functions, thus improving the previously best-known bounds of $2n+1$ and $9$.
\end{abstract}

\section{Introduction} \label{sec:intro}
We consider the following truthful heterogeneous facility location problem with max-variant cost: There is a set $N$ of $n$ {\em agents} with private positions on the line of real numbers ($x_i$ for agent $i$), and public {\em approval preferences} $p_i \in \{0,1\}^2$ over two facilities $F_1$ and $F_2$ with $p_{i1}+p_{i2}\geq 1$ (i.e., each agent approves at least one facility). Let $N_j$ be the set of agents that approve facility $j \in [2]$. Clearly, $N_1$ and $N_2$ need not be disjoint since there might be agents that approve both facilities; we denote by $N_1 \setminus N_2$ the set of agents that approve only $F_1$, $N_2 \setminus N_1$ the set of agents that approve only $F_2$, and $N_1 \cap N_2$ the set of agents that approve both facilities. 

There is also a finite set $C$ of given {\em candidate locations}, where the facilities can be placed. 
For any agent $i$ and point $x$ of the line, let $d(i,x) = |x_i - x|$ denote the distance between the position of $i$ and $x$.
Given a {\em feasible solution} $\by = (y_1,y_2)$ consisting of the location $y_1$, where $F_1$ is placed, and the location $y_2 \neq y_1$, where $F_2$ is placed, the {\em individual cost} of an agent $i$ is her distance to the {\em farthest} facility among the ones she approves (hence, the term {\em max-variant cost} in the title of the problem), that is, 
$$\cost_i(\by) = \max_{j: p_{ij} = 1} d(i,y_j),$$
The {\em social cost} of a solution $\by = (y_1,y_2)$ is the total cost of all agents:
\begin{align*}
    \SC(\by) = \sum_{i \in N} \cost_i(\by).
\end{align*}
The {\em max cost} is the maximum cost over all agents:
\begin{align*}
    \MC(\by) = \max_{i \in N} \cost_i(\by).
\end{align*}
A {\em mechanism} takes as input the positions of the agents and, together with the public information about their preferences, decides a feasible solution, which we will typically denote by $\bw=(w_1,w_2)$ in our analysis. For any $f\in \{\SC,\MC\}$, the {\em approximation ratio} of a mechanism in terms of $f$ is the worst-case (over all possible instances) of the ratio $\frac{f(\bw)}{\min_\by f(\by)}$. Our goal is to design mechanisms that achieve a low {\em approximation ratio} in terms of both the social cost and the max cost simultaneously, while at the same time being {\em strategyproof}, that is, do not provide incentive to the agents to misreport their private positions and decrease their individual cost. 

\citet{Zhao2023constrained} were the first to consider this constrained heterogeneous facility location problem with max-variant cost; a version of the problem where the facilities can be placed at the same locations was previously studied by \citet{chen2020max} under different informational assumptions (known positions and private preferences).  \citeauthor{Zhao2023constrained} showed that there are strategyproof mechanisms which achieve approximation ratios of $2n+1$ and $9$ for the social cost and max cost objectives, respectively. For both objectives, the mechanisms of \citet{Zhao2023constrained} switch between cases depending on whether there is at least one agent that approves both facilities or not. When there is such an agent, the facilities are placed at the candidate locations that are closest to a designated agent (the median for the social cost and the leftmost for the max cost). When no agent approves both facilities, each facility is separately placed to the available candidate location that is closest to a designated agent among the ones that approve it (again, median for the social cost or leftmost for the max cost); for the social cost, the facility that is placed first is the one that the majority of agents approve. 

\subsection*{Our contribution}
In this work, we design a different strategyproof mechanism and show that it simultaneously achieves an approximation ratio of at most $11$ for the social cost and at most $5$ for the maximum cost, thus improving both bounds at the same time; for the social cost, the improvement is significant as the approximation decreases from linear to constant. Our mechanism works as follows. For simplicity, suppose that $F_1$ is the facility that most agents approve. The mechanism switches between two cases depending on whether the number of agents that approve $F_1$ is at least the number of agents that approve both facilities. If this is indeed the case, $F_1$ placed at the candidate location closest to the median among the set of agents that approve only it, while $F_2$ is placed at the available candidate location closest to the median agent among the ones that approve it (including the agents that approve both). Otherwise, the two facilities are placed at the two candidate locations that are closest to the median agent among those that approve both facilities. Our analysis uses a combination of simple tools (such as the triangle inequality and the fact that the median of a set of points minimizes their total distance to any point) and worst-case characterization which reveals the properties that worst-case instances (where the approximation ratio is maximized) have. These results are presented in Section~\ref{sec:results}. 

In Section~\ref{sec:connection}, we present an interesting connection between the max-variant truthful heterogeneous two-facility location problem (which is the main focus of this paper) and the sum-variant considered by \citet{kanellopoulos2023}. We show that any mechanism that is strategyproof in both variants and achieves an approximation ratio of at most $\rho$ in terms of the social cost or the maximum cost in one of the variants, also achieves an approximation ratio of at most $2\rho$ in the other variant. This allows us to show that a constant approximation ratio (at most $22$) can be achieved for the social cost in the sum-variant, which is the first constant bound for this version of the problem as well; \citet{kanellopoulos2023} previously claimed an upper bound of $3$ but with a mechanism that turned out to not be strategyproof. 

\subsection*{Other related work}
The seminal paper of \citet{procaccia09approximate} initiated the study of facility location problems under the prism of approximate mechanism design without money. Since their work, a plethora of different variants have been studied in the literature under various assumptions, always aiming to design strategyproof mechanisms with an as low approximation ratio as possible in terms of some social objective. Indicatively, different variants have been considered depending on the number of facilities to locate~\citep{Lu2010two-facility,fotakis2014two}, whether the facilities are desirable or obnoxious~\citep{cheng2013obnoxious}, whether the preferences of the agents are homogeneous or heterogeneous~\citep{feigenbaum2015hybrid,serafino2016,kanellopoulos2021discrete,deligkas2023limited,li2020constant,chen2020max}, and whether there are constraints about where the facilities can be placed~\citep{feldman2016voting,kanellopoulos2023,Xu2021minimum} in combination with min or sum individual costs. For a detailed exposition of the truthful facility location literature, we refer the interest reader to the survey of \citet{fl-survey}. 

\subsection*{Some useful notation and observations}
We will extensively use the {\em triangle inequality} stating that $d(x,y) \leq d(x,z) + d(z,y)$ for any three points $x,y,z$ of the line. We will sometimes use directly the following version of the triangle inequality (which can be shown by applying the classic one twice): $d(i,x) \leq d(i,y) + d(j,y) + d(j,x)$ for any two agents $i,j$ and locations $x,y$. 

When the objective (either the social cost or the maximum cost) is clear from context, we will denote by $\bo = (o_1,o_2)$ an optimal solution, that is, $\bo$ is a feasible solution that minimizes the objective. In addition, for any agent $i$, we will denote by $o(i)$ the location of the facility that determines the cost of $i$ in $\bo = (o_1,o_2)$. Observe that, for any $j \in [2]$, if $i \in N_j$, then $d(i,o_j) \leq d(i,o(i)) = \cost_i(\bo)$; this follows by the fact that either $o_j = o(i)$ or $o_j$ is closer to $i$ than $o(i) = o_{3-j}$. Finally, we will denote by $t(i)$ and $s(i)$ the candidate locations that are closest and second-closest to the position of agent $i$. Observe that, by definition, $d(i,t(i)) \leq d(i,o(i))$. 

%%%%%%%%%%%
%%%%%%%%%%%
\section{Results} \label{sec:results}
%%%%%%%%%%%
%%%%%%%%%%%
We present and analyze the mechanism for the case where $|N_1| \geq |N_2|$; for the case $|N_2| > |N_1|$, it suffices to swap $1$ and $2$ in the description below. 
\begin{itemize}
    \item \textbf{(Case 1)} If $|N_1 \setminus N_2| \geq |N_1 \cap N_2|$, then place $F_1$ at the location $t(m_1)$ that is closest to the median agent $m_1 \in N_1 \setminus N_2$, and place $F_2$ at the location that is closest to the median agent $m_2 \in N_2$ from the set of available locations; so, either $t(m_2)$ if $t(m_2) \neq t(m_1)$, or $s(m_2)$ if $t(m_2) = t(m_1)$.

    \item \textbf{(Case 2)} Otherwise, place $F_1$ at $t(m_{12})$ and $F_2$ at $s(m_{12})$, where $m_{12}$ is the median agent of $N_1 \cap N_2$.
\end{itemize}
We refer to this mechanism as {\sc Conditional-Median}.

\begin{theorem} \label{thm:sp}
{\sc Conditional-Median} is strategyproof. 
\end{theorem}

\begin{proof}
First observe that the two cases considered by the mechanism only depend on the preferences of the agents, which are publicly known. Hence, no agent can force the mechanism to go from (Case 1) to (Case 2) or vice versa by misreporting. We now discuss each case separately. 
\begin{itemize}
    \item \textbf{(Case 1)} Since $F_1$ is places at $t(m_1)$, the cost of $m_1$ is minimized. To change the location of $F_1$, an agent $i \in (N_1 \setminus N_2) \setminus \{m_1\}$ would have to misreport a position such that the median of $N_1 \setminus N_2$ changes, which would either not change the location of $F_1$ or move it farther away. Clearly, no agent of $N_2$ can affect the location of $F_1$, and, for the same reason, has no incentive to misreport. In particular, given the location of $F_1$, the cost of $m_2$ is minimized, and any other agent $i \in (N_2 \setminus N_1) \setminus \{m_2\}$ would have to change the median of $N_2$, which would either not change the location of $F_2$ or move it farther away. 

    \item \textbf{(Case 2)} No agent of $(N_1 \setminus N_2) \cup (N_2 \setminus N_1)$ can affect the outcome of the mechanism in this case. The cost of $m_{12}$ is clearly minimized, and no other agent $i \in (N_1 \cap N_2)\setminus \{m_{12}\}$ can misreport; similarly to (Case 1), $i$ would have to change the median of $N_1 \cap N_2$, which would either not change the location of the facilities or move them farther away. 
\end{itemize}
Consequently, in any of the two cases, no agent has incentive to misreport, and thus the mechanism is strategyproof. 
\end{proof}

\begin{theorem} \label{thm:SC}
For the social cost, the approximation ratio of {\sc Conditional-Median} is at most $11$. 
\end{theorem}

\begin{proof}
Let $\bw = (w_1,w_2)$ be the solution computed by the mechanism, and $\bo = (o_1,o_2)$ be an optimal solution. 
We consider each of the two cases of the mechanism separately. In each case, we combine known properties related to the median of a set of points on the line, the triangle inequality, and some with properties of worst-case instances (which achieve the largest possible approximation ratio).

\medskip

\noindent 
{\bf (Case 1).}
We consider two subcases depending on whether $m_2$ is closer to $w_2$ than to $o_2$. Observe that this is definitely true when $t(m_1) \neq t(m_2)$ and might be true when $t(m_1) = t(m_2)$. If it is not true when $t(m_1) = t(m_2)$, then, since $w_2 = s(m_2)$, we have more information about the structure of the instance, in particular, we have that $w_1=o_2$.
For any $j \in [2]$, let $S_j$ be the subset of $N_1 \cap N_2$ that includes agents whose cost in $\bw$ is determined by $w_j$.

\medskip

\noindent 
{\bf (Case 1.1): $d(m_2,w_2) \leq d(m_2,o_2)$.}  
The social cost of $\bw$ can be written as follows:
\begin{align*}
\SC(\bw) 
&= \sum_{i \in N_1 \setminus N_2} d(i,w_1) + \sum_{i \in S_1} d(i,w_1) + \sum_{i \in (N_2 \setminus N_1) \cup S_2} d(i,w_2) \\
&\leq \sum_{i \in N_1 \setminus N_2} d(i,w_1) + \sum_{i \in S_1} d(i,w_1) + \sum_{i \in N_2} d(i,w_2)
\end{align*}
For the agents in $N_1 \setminus N_2$, using the triangle inequality, the fact that $d(m_1,w_1) \leq d(m_1,o_1)$, $o_1 = o(i)$ for every $i \in N_1 \setminus N_2$, and the property of $m_1$ which minimizes the total distance of all agents in $N_1 \setminus N_2$ from any point on the line, including $o_1$, we obtain
\begin{align}
     \sum_{i \in N_1 \setminus N_2} d(i,w_1) 
     &\leq \sum_{i \in N_1 \setminus N_2} \bigg( d(i,m_1) + d(m_1,w_1) \bigg) \nonumber \\
     &\leq \sum_{i \in N_1 \setminus N_2} \bigg( d(i,m_1) + d(m_1,o(i)) \bigg) \nonumber \\
     &\leq \sum_{i \in N_1 \setminus N_2} \bigg( 2\cdot d(i,m_1) + d(i,o(i)) \bigg) \nonumber \\
     &\leq 3 \cdot \sum_{i \in N_1 \setminus N_2} d(i,o(i)). \label{eq:case1.1:N1}
\end{align}
For the agents in $S_1$, we have that $|S_1| \leq |N_1 \cap N_2| \leq |N_1 \setminus N_2|$. In addition, $d(i, o_1) \leq d(i,o(i))$ for every agent $i \in S_1$, $d(m_1,w_1) \leq d(m_1,o_1)$, and $o_1 = o(i)$ for any agent $i \in N_1 \setminus N_2$. So, using the triangle inequality and the property of $m_1$ which minimizes the total distance of all agents in $N_1 \setminus N_2$ from any point on the line, including $o_1$, we obtain
\begin{align}
    \sum_{i \in S_1} d(i,w_1) 
    &\leq \sum_{i \in S_1} \bigg( d(i,o_1) + d(m_1,o_1) + d(m_1,w_1) \bigg) \nonumber \\
    &\leq \sum_{i \in S_1} d(i,o(i)) + \sum_{i \in N_1 \setminus N_2} 2\cdot d(m_1,o_1) \nonumber \\
    &\leq \sum_{i \in S_1} d(i,o(i)) + \sum_{i \in N_1 \setminus N_2} \bigg( 2\cdot d(i,m_1) + 2\cdot d(i,o(i)) \bigg) \nonumber \\
    &\leq \sum_{i \in S_1} d(i,o(i)) + 4 \cdot \sum_{i \in N_1 \setminus N_2} d(i,o(i)). \label{eq:case1.1:S1}
\end{align}
For the agents in $N_2$, using the triangle inequality, the fact that $d(m_2,w_2) \leq d(m_2,o(i))$ for every $i \in N_2$, and the property of $m_2$ which minimizes the total distance of all agents in $N_2$ from any point on the line, including $o_2$,
we obtain
\begin{align}
     \sum_{i \in N_2 } d(i,w_2) 
     &\leq \sum_{i \in N_2} \bigg( d(i,m_2) + d(m_2,w_2) \bigg) \nonumber \\
     &\leq \sum_{i \in N_2} \bigg( d(i,m_2) + d(m_2,o(i)) \bigg) \nonumber \\
     &\leq \sum_{i \in N_2} \bigg( 2\cdot d(i,m_2) + d(i,o(i)) \bigg) \nonumber \\
     &\leq 3 \cdot \sum_{i \in N_2} d(i,o(i)) \label{eq:case1.1:N2}
\end{align}
Putting everything together, and using the facts that $N = (N_1 \setminus N_2) \cup N_2$ and $(N_1 \setminus N_2) \cup S_1 \subseteq N_1 \subseteq N$, we have
\begin{align*}
\SC(\bw) &\leq 7\cdot \sum_{i \in N_1 \setminus N_2} d(i,o(i)) 
+ \sum_{i \in S_1} d(i,o(i)) + 3 \cdot \sum_{i \in N_2} d(i,o(i)) \\
&\leq 7 \cdot \SC(\bo),
\end{align*}
that is, the approximation ratio is at most $7$. 

\medskip

\noindent 
{\bf (Case 1.2): $d(m_2,w_2) > d(m_2,o_2)$.}
First observe that if $w_1 = t(m_1)$ and $w_2 = t(m_2)$, then it would have to be the case that $d(m_2,w_2) \leq d(m_2,o_2)$, which is already captured by (Case 1.1). So, it must be the case that $w_1 = t(m_1) = t(m_2)$ and $w_2 = s(m_2)$, which, in combination with the fact that $d(m_2,w_2) > d(m_2,o_2)$, implies that $w_1 = o_2$. 

Since $m_1$ is closer to $w_1=o_2$ than to $o_1$, the same must be true for at least half of the agents in $N_1 \setminus N_2$, which implies the following lower bound on the optimal social cost:
\begin{align*}
    \SC(\bo) \geq \frac{|N_1 \setminus N_2|}{2} \cdot \frac{d(o_1,o_2)}{2} = \frac{|N_1 \setminus N_2|}{4} \cdot d(o_1,o_2).
\end{align*}

Now, Inequalities \eqref{eq:case1.1:N1} and \eqref{eq:case1.1:S1} are still true since $w_1 = t(m_1)$. We will use \eqref{eq:case1.1:N1} in this case as well to bound the contribution of the agents in $N_1 \setminus N_2$ to the social cost of $\bw$, but for the agents in $S_1$ we will use a different bound obtained by applying the triangle inequality once, as follows:
\begin{align} \label{eq:case1.2:S1}
    \sum_{i \in S_1} d(i,w_1) &\leq \sum_{i \in S_1} \bigg( d(i,o_1) + d(o_1,w_1) \bigg) \nonumber \\
    &= \sum_{i \in S_1} d(i,o_1) + |S_1| \cdot d(o_1,o_2). 
\end{align}
We will now bound the contribution of the agents in $(N_2 \setminus N_1) \cup S_2$ by using the triangle inequality, the fact that $d(m_2,w_2) \leq d(m_2,o_1)$ (which is true since $t(m_2)=w_1=o_2$ and $s(m_2)=w_2$), the fact that $(N_2 \setminus N_1) \cup S_2 \subseteq N_2$, and the fact that $m_2$ minimizes the total distance of all agents in $N_2$ from any point of the line, including $o_2$. We have
\begin{align} \label{eq:case1.2:N2}
\sum_{i \in (N_2 \setminus N_1) \cup S_2} d(i,w_2) 
&\leq \sum_{i \in (N_2 \setminus N_1) \cup S_2} \bigg( d(i,m_2) + d(m_2,w_2) \bigg) \nonumber\\
&\leq \sum_{i \in (N_2 \setminus N_1) \cup S_2} \bigg( d(i,m_2) + d(m_2,o_1) \bigg) \nonumber\\
&\leq \sum_{i \in (N_2 \setminus N_1) \cup S_2} \bigg( d(i,m_2) + d(m_2,o_2) + d(o_1,o_2) \bigg) \nonumber\\
&\leq 2 \cdot \sum_{i \in N_2} d(i,m_2) +  \sum_{i \in (N_2 \setminus N_1) \cup S_2} d(i,o_2) + \bigg(|N_2 \setminus N_1|+ |S_2|\bigg) \cdot d(o_1,o_2) \nonumber\\
&\leq 2 \cdot \sum_{i \in N_2} d(i,o(i)) + \sum_{i \in (N_2 \setminus N_1) \cup S_2} d(i,o_2) + \bigg(|N_2 \setminus N_1|+ |S_2|\bigg) \cdot d(o_1,o_2).
\end{align}
Combining Inequalities \eqref{eq:case1.1:N1}, \eqref{eq:case1.2:S1} and \eqref{eq:case1.2:N2}, the fact that $(N_2 \setminus N_1) \cup S_2 \cup S_1 = N_2$, the fact that $|N_2| \leq 2 \cdot |N_1 \setminus N_2|$ (which is true since $|N_1 \setminus N_2| \geq |N_1 \cap N_2|$ and $|N_1 \setminus N_2| \geq |N_2 \setminus N_1|$), and the lower bound on $\SC(\bo)$, we finally have:
\begin{align*}
    \SC(\bw) 
    &\leq 3 \cdot \SC(\bo) + |N_2| \cdot d(o_1,o_2) \\
    &\leq 3 \cdot \SC(\bo) + 2\cdot |N_1 \setminus N_2| \cdot d(o_1,o_2) \\
    &\leq 11 \cdot \SC(\bo),  
\end{align*}
that is, the approximation ratio is at most $11$.

\medskip 

\noindent 
{\bf (Case 2)}
Here, we have that $w_1 = t(m_{12})$ and $w_2 = s(m_{12})$. Since the decision does not depend on the positions of the agents in $(N_1 \setminus N_2) \cup (N_2 \setminus N_1)$, these agents can be moved anywhere without affecting the outcome of the mechanism. Hence, in a worst-case instance, to maximize the approximation ratio, all these agent can be positioned at exactly the optimal locations of the two facilities, namely, all agents of $N_1 \setminus N_2$ are positioned exactly at $o_1$ and all agents of $N_2 \setminus N_1$ are positioned exactly at $o_2$. 

Since $|N_1 \setminus N_2| \leq |N_1 \cap N_2|$, in a worst-case instance, the contribution of the agents of $N_1 \setminus N_2$ to the social cost of the solution $\bw$ computed by the mechanism is
\begin{align*}
    \sum_{i \in N_1 \setminus N_2} d(i,w_1) 
    = \sum_{i \in N_1 \setminus N_2} d(o_1,w_1) 
    \leq \sum_{i \in N_1 \cap N_2} d(o_1,w_1) 
\end{align*}
By the triangle inequality, we further have that
\begin{align}
    \sum_{i \in N_1 \setminus N_2} d(i,w_1) 
    &\leq \sum_{i \in N_1 \cap N_2} \bigg( d(i,o_1) + d(i,w_1) \bigg) \nonumber \\
    &\leq \sum_{i \in N_1 \cap N_2} \bigg( d(i,o(i)) + d(i,w(i)) \bigg), \label{eq:case2:N1agents}
\end{align}
where $w(i)$ is the location of the facility that determines the cost of agent $i$ in $\bw$. 

Since $|N_2| \leq |N_1|$, we also have that $|N_2 \setminus N_1| \leq |N_1 \setminus N_2| \leq |N_1 \cap N_2|$. Consequently, similarly to above, in a worst-case instance, the contribution of the agents of $N_2 \setminus N_1$ to the social cost of $\bw$ can be upper-bounded as follows:
\begin{align}
    \sum_{i \in N_2 \setminus N_1} d(i,w_2) 
    &\leq \sum_{i \in N_1 \cap N_2} \bigg( d(i,o(i)) + d(i,w(i)) \bigg) \label{eq:case2:N2agents}
\end{align}

The social cost of $\bw$ in a worst-case instance can now be written as follows:
\begin{align}
\SC(\bw) 
&= \sum_{i \in N_1 \setminus N_2} d(i,w_1) + \sum_{i \in N_2 \setminus N_1} d(i,w_2) + \sum_{i \in N_1 \cap N_2} d(i,w(i)) \nonumber \\
&\leq 2 \cdot \sum_{i \in N_1 \cap N_2} d(i,o(i)) + 3 \cdot \sum_{i \in N_1 \cap N_2} d(i,w(i)). \label{eq:case2:SC}
\end{align}
For any agent $i \in N_1 \cap N_2$, by the triangle inequality, the fact that $d(m_{12},w_1) \leq d(m_{12},w_2) \leq d(m_{12},o(m_{12}))$, and the fact that $d(i,o(m_{12})) \leq d(i,o(i))$, we have 
\begin{align*}
    d(i,w(i)) &\leq d(i,m_{12}) + d(m_{12},w(i)) \\
    &\leq d(i,m_{12}) + d(m_{12},o(m_{12})) \\
    &\leq 2 d(i,m_{12}) + d(i,o(m_{12})) \\
    &\leq 2 d(i,m_{12}) + d(i,o(i)) 
\end{align*}
Hence, 
\begin{align*}
\sum_{i \in N_1 \cap N_2} d(i,w(i)) \leq 2 \cdot \sum_{i \in N_1 \cap N_2} d(i,m_{12}) + \sum_{i \in N_1 \cap N_2} 
 d(i,o(i)).
\end{align*}
Since $m_{12}$ minimizes the total distance of the agents in $N_1 \cap N_2$ from any point of the line, including $o_1$, and the fact that $d(i,o_1) \leq d(i,o(i))$ for any agent $i \in N_1 \cap N_2$, we have
\begin{align}
\sum_{i \in N_1 \cap N_2} d(i,w(i)) 
&\leq 2 \cdot \sum_{i \in N_1 \cap N_2} d(i,o_1) + \sum_{i \in N_1 \cap N_2} 
 d(i,o(i)) \nonumber \\
&\leq 3 \cdot \sum_{i \in N_1 \cap N_2} d(i,o(i)). \label{eq:case2:N12}
\end{align}
Replacing \eqref{eq:case2:N12} to \eqref{eq:case2:SC} and using the obvious fact that $N_1 \cap N_2 \subseteq N$, we finally obtain that
\begin{align*}
    \SC(\bw) \leq 11 \cdot \sum_{i \in N_1 \cap N_2} d(i,o(i)) \leq 11 \cdot \SC(\bo),
\end{align*}
that is, the approximation ratio is at most $11$ in this case, completing the proof. 
\end{proof}

We also show that the bound of $11$ on the approximation ratio of the mechanism in terms of the social cost is tight with the following example. 

\begin{example}
Let $\varepsilon>0$ be an infinitesimal and consider an instance with the following four candidate locations: $\{0, \varepsilon, 1, 1+\varepsilon\}$. 
There are $n/3$ agents that approve $F_1$ positioned at $0$, $n/3$ agents that approve $F_2$ positioned at $0$, $n/6$ agents that approve both facilities positioned at $0$, and another $n/6+1$ agents that approve both facilities all positioned $1/2+2\varepsilon$. Clearly, this instance is captured by (Case 2), and thus the mechanism places the two facilities at the candidate locations closest to the median of $N_1 \cap N_2$ who is positioned at $1$. Hence, $w_1 = 1$ and $w_2 = 1+\varepsilon$, for a social cost of approximately $2n/3 + n/6 + n/12 = 11n/12$. However, the optimal solution is to place the two facilities at $0$ and $\varepsilon$ for a social cost of approximately $n/12$, leading to an approximation ratio of $11$. \hfill $\qed$
\end{example}

We now turn our attention to the max cost objective. 

\begin{theorem}\label{thm:MC}
For the maximum cost, the approximation ratio of {\sc Conditional-Median} is at most $5$. 
\end{theorem}

\begin{proof}
Let $i$ be the agent with the maximum individual cost over all agents for the solution $\bw = (w_1,w_2)$ computed by the mechanism. Also, let $\bo = (o_1,o_2)$ be an optimal solution.
We consider the following two cases depending on which facility determines the individual cost of $i$. 

\medskip
\noindent
{\bf Case (a): The individual cost of $i$ in $\bw$ is determined by $F_1$.} \ \\
Clearly, since the cost of $i$ is determined by $F_1$, $i \in N_1$, and thus $d(i,o_1) \leq d(i,o(i))$. Let $m$ be the agent that determines the location of $F_1$; so, $m$ is either $m_1$ in (Case 1) or $m_{12}$ in (Case 2). Observe that $w_1$ is the candidate location closest to $m$, and thus $d(m,w_1) \leq d(m,o(m))$. Also, since $m \in N_1$, $d(m,o_1) \leq d(m,o(m))$. 
Hence, by the triangle inequality, we have that
\begin{align*}
    \MC(\bw) = d(i,w_1) 
    &\leq d(i,o_1) + d(m,o_1) + d(m,w_1) \\
    &\leq d(i,o(i)) + 2 \cdot d(m, o(m))\\
    &\leq 3 \cdot \MC(\bo).
\end{align*}

\medskip
\noindent
{\bf Case (b): The individual cost of $i$ in $\bw$ is determined by $F_2$.} \ \\
If we are in (Case 1) and $w_2 = t(m_2)$, then we can get an upper bound of $3$ similarly to Case (a). We can also derive an upper bound of $3$ if we are in (Case 2), where $w_2 = s(m_{12})$, since this implies that for some of $o \in \{o_1,o_2\}$, $d(m_{12},w_2) \leq d(m,o) \leq d(m,o(m))$; using this and the fact that $d(m,o_2) \leq d(m,o(m))$, we can apply the triangle inequality as in Case (a) to get the bound of $3$. 

So, let us now assume that we are in (Case 1) and $w_2 = s(m_2)$, which implies that $t(m_2) = w_1 = t(m_1)$. 
Clearly, if $m_2$ is (weakly) closer to $w_2$ than to $o_2$, then the same argument as in Case (a) can again lead to an upper bound of $3$. So, we can further assume that $t(m_2) = w_1 = t(m_1) = o_2$. Due to this, $m_2$ is (weakly) closer to $w_2 = s(m_2)$ than to $o_1$. Therefore, using repeatedly the triangle inequality and this relation between these points, we obtain
\begin{align*}
    \MC(\bw) = d(i,w_2) 
        &\leq d(i,o_2) + d(m_2,o_2) + d(m_2,w_2) \\
        &\leq d(i,o(i)) + d(m_2,o(m_2)) + d(m_2,o_1) \\
        &\leq 2 \cdot \MC(\bo) + d(m_2,o_1) \\
        &\leq 2 \cdot \MC(\bo) + d(m_2,o_2) + d(m_1,o_2) + d(m_1,o_1) \\
        &\leq 2 \cdot \MC(\bo) + d(m_2,o(m_2)) + d(m_1,t(m_1)) + d(m_1,o(m_1)) \\
        &\leq 5 \cdot \MC(\bo), 
\end{align*}
and the proof is complete.
\end{proof}

We also show that the bound of $5$ shown in the above theorem is tight. 

\begin{example}
Let $\varepsilon>0$ be an infinitesimal and consider an instance with the following three candidate locations: $\{0, 2, 6\}$. 
There are three agents that approve $F_1$ all positioned at $1+\varepsilon$, one agent that approves $F_2$ positioned at $1$, and another two agents that approve $F_2$ positioned at $3+\varepsilon$. This instance is captured by (Case 1), and thus the mechanism places $F_1$ at $w_1 = t(m_1) = 2$, and since $t(m_2)=2$ as well, $F_2$ is placed at $w_2 = s(m_2) = 6$. The max cost of this solution is $5$, realized by the agent that approves $F_2$ positioned at $1$. However, the optimal solution is to place $F_1$ at $0$ and $F_2$ at $2$ for a max cost of approximately $1$, leading to an approximation ratio of $5$.
\hfill $\qed$
\end{example}

%%%%%
%%%%%
\section{A connection between the max- and the sum-variant} \label{sec:connection}
In this short section, we discuss an interesting connection between the max-variant heterogeneous two-facility location problem that is the focus of this paper and the {\em sum-variant} that has been studied by \citet{kanellopoulos2023}. In the sum-variant, the individual cost of an agent $i$ for a feasible solution $\by$ is her distance to all the facilities she approves, in particular, 
\begin{align*}
    \cost_i^\text{sum}(\by) = \sum_{j \in [2]} p_{ij}\cdot d(i,y_j).
\end{align*}
In contrast, recall that, in the max-variant, the individual cost of an agent $i$ is her distance to the farthest facility among the ones she approves, i.e., 
\begin{align*}
    \cost_i^\text{max}(\by) = \max_{j: p_{ij}=1} d(i,y_j).
\end{align*}
It is not hard to observe that $\frac12 \cdot \cost_i^\text{sum}(\by) \leq \cost_i^\text{max}(\by) \leq \cost_i^\text{sum}(\by) \leq 2 \cdot \cost_i^\text{max}(\by)$ for any agent $i$, which implies the following statement. 

\begin{theorem}\label{thm:connection}
Let $M$ be any mechanism that is strategyproof for the max-variant and the sum-variant. 
If the approximation ratio of $M$ in terms of the social cost or the maximum cost is at most $\rho \geq 1$ in one variant, then it is at most $2\cdot \rho$ in the other variant. 
\end{theorem}

\begin{proof}
Let $\bw$ be the solution computed by $M$. Also, let $\bo^\text{max}$ and $\bo^\text{sum}$ be the optimal solutions in terms of the objective $f \in \{\SC,\MC\}$ for the max- and the sum-variant, respectively. For convenience, we write $f^\text{max}(\by)$ and $f^\text{sum}(\by)$ for the objective value of solution $\by$ in terms of $f$ in the max- and the sum-variant, respectively.
Due to the relation between the individual costs of the agents in the two variants mentioned above and the fact that $f$ is either the sum of the costs or the cost of some agent, we have that $f^\text{sum}(\bw) \leq 2\cdot f^\text{max}(\bw)$ and $f^\text{sum}(\bo^\text{sum}) \geq f^\text{max}(\bo^\text{sum})$. Combined with the fact that $f^\text{max}(\bo^\text{sum}) \geq f^\text{max}(\bo^\text{max})$, we obtain
\begin{align*}
    \frac{f^\text{sum}(\bw)}{f^\text{sum}(\bo^\text{sum})} \leq \frac{2\cdot f^\text{max}(\bw)}{f^\text{max}(\bo^\text{sum})}
    \leq 2\cdot \frac{f^\text{max}(\bw)}{f^\text{max}(\bo^\text{max})} \leq 2\cdot \rho.
\end{align*}
Similarly, we have that $f^\text{max}(\bw) \leq f^\text{sum}(\bw)$ and $f^\text{max}(\bo^\text{max}) \geq \frac12 \cdot f^\text{sum}(\bo^\text{max})$. Combined with the fact that $f^\text{sum}(\bo^\text{max}) \geq f^\text{sum}(\bo^\text{sum})$, we obtain
\begin{align*}
    \frac{f^\text{max}(\bw)}{f^\text{max}(\bo^\text{max})} \leq \frac{f^\text{sum}(\bw)}{\frac12 \cdot f^\text{sum}(\bo^\text{max})}
    \leq 2\cdot \frac{f^\text{sum}(\bw)}{f^\text{sum}(\bo^\text{sum})} \leq 2\cdot \rho,
\end{align*}
as desired.
\end{proof}

It is not hard to observe that {\sc Conditional-Median} is strategyproof not only in the max-variant but also in the sum-variant (in particular, the arguments in the proof of Theorem~\ref{thm:sp} follow through for the sum-variant). 
Consequently, by Theorems~\ref{thm:SC}, \ref{thm:MC} and \ref{thm:connection}, {\sc Conditional-Median} achieves constant approximation ratio in terms of the social cost and the maximum cost in the sum-variant. As already mentioned, this fixes a result of \citet{kanellopoulos2023} who claimed an upper bound of $3$ for the social cost with a mechanism that turned out to not be strategyproof.

\begin{corollary}
In the sum-variant, the approximation ratio of {\sc Conditional-Median} is at most $22$ in terms of the social cost and at most $10$ in terms of the maximum cost. 
\end{corollary}

%%%%%
%%%%%

%%%%%
%%%%%
\section{Open problems} \label{sec:open}
%%%%%
%%%%%
In this paper, we focused on the truthful constrained heterogeneous facility location problem with max-variant cost. We showed that constant bounds of $11$ and $5$ on the approximation ratio in terms of the social cost and maximum cost, respectively, can be achieved {\em simultaneously} by a rather simple and natural strategyproof mechanism, thus improving upon the previous best-known bounds. The most direct open question is whether improved approximations can be simultaneously achieved in terms of these two objectives. Building on this idea (which has recently been considered by \citet{han2023simultaneous} for non-strategic voting and single-facility location), one could also consider other objectives that interpolate between the social cost and the max cost, such as the sum of the $k \in [n]$ max individual agent costs; observe that this objective coincides with the social cost for $k=n$ and with the max cost for $k=1$. Another direction would be to consider different assumptions, such as when both the positions and the preferences of the agents are private. 

\bibliographystyle{plainnat}
\bibliography{max-FL-references}

\end{document}